\documentclass[12pt]{article}

\usepackage{mathpazo}
\usepackage{times}

\usepackage{amsthm}
\usepackage{amsmath}  
\usepackage{amsfonts}
\usepackage{amssymb}  
\usepackage{mathrsfs} 

\usepackage{cite}     
 
\usepackage{upref}

\usepackage{comment}  
\usepackage{fullpage}

\title{\Huge\bf$\,$\\[-4.00ex]
{Lower Bound on the Redundancy of~PIR~Codes}\\[1.00ex]}

\author{\large Sankeerth Rao and Alexander Vardy}

\date{}

\newtheorem{theorem}{Theorem}

\newtheorem{lemma}[theorem]{Lemma}

\newtheorem{defn}{Definition$\!$}
\newenvironment{definition}{\begin{defn}}{\end{defn}}

\newcommand{\Tref}[1]{The\-o\-rem\,\ref{#1}}

\newcommand{\Lref}[1]{Lem\-ma\,\ref{#1}}
\newcommand{\Cref}[1]{Co\-ro\-lla\-ry\,\ref{#1}}
\newcommand{\Dref}[1]{De\-fin\-i\-tion\,\ref{#1}}

\newcommand{\dfn}{\bfseries\itshape}

\newcommand{\be}[1]{\begin{equation}\label{#1}}
\newcommand{\ee}{\end{equation}} 
\newcommand{\eq}[1]{(\ref{#1})}

\renewcommand{\Bbb}{\mathbb}
\newcommand{\C}{{\Bbb C}}
\newcommand{\F}{{\Bbb F}}

\newcommand{\Span}[1]{{\left\langle {#1} \right\rangle}}

\newcommand{\deff}{\mbox{$\stackrel{\rm def}{=}$}}

\newcommand{\zero}{{\mathbf 0}}
\newcommand{\al}{\alpha}

\newcommand{\shalf}{\mbox{\raisebox{.8mm}{\footnotesize $\scriptstyle 1$}
\footnotesize$\!\!\! / \!\!\!$ \raisebox{-.8mm}{\footnotesize
$\scriptstyle 2$}}}

\newcommand{\sfourth}{\mbox{\raisebox{.8mm}{\footnotesize $\scriptstyle 1$}
\footnotesize$\!\!\! / \!\!\!$ \raisebox{-.8mm}{\footnotesize
$\scriptstyle 4$}}}

\renewcommand{\le}{\leqslant}

\renewcommand{\ge}{\geqslant}

\DeclareMathAlphabet{\mathbfsl}{OT1}{ppl}{b}{it} 

\newcommand{\aaa}{\mathbfsl{a}} 
\newcommand{\bbb}{\mathbfsl{b}} 
 
\newcommand{\eee}{\mathbfsl{e}} 
 
\newcommand{\uuu}{\mathbfsl{u}} 
\newcommand{\vvv}{\mathbfsl{v}}
 
\newcommand{\xxx}{\mathbfsl{x}}

\newcommand{\cA}{{\cal A}}

\newcommand{\cP}{{\cal P}}

\makeatletter

\gdef\@punct{.\ \ }  
\def\@sect#1#2#3#4#5#6[#7]#8{%
  \ifnum #2>\c@secnumdepth
     \def\@svsec{}
  \else
     \refstepcounter{#1}\edef\@svsec{%
     \ifnum #2>0{{\csname the#1\endcsname}}.\fi%
    \hskip .5em}
  \fi
  \@tempskipa #5\relax
  \ifdim \@tempskipa>\z@
     \begingroup #6\relax
       \@hangfrom{\hskip #3\relax\@svsec}{\interlinepenalty \@M #8\par}
     \endgroup
     \csname #1mark\endcsname{#7}
     \addcontentsline{toc}{#1}{\ifnum #2>\c@secnumdepth\else
          \protect\numberline{\csname the#1\endcsname}\fi#7}
  \else
     \def\@svsechd{#6\hskip #3\@svsec #8\@punct\csname #1mark\endcsname{#7}
     \addcontentsline{toc}{#1}{\ifnum #2>\c@secnumdepth \else
          \protect\numberline{\csname the#1\endcsname}\fi#7}}
  \fi
  \@xsect{#5}}

\def\@ssect#1#2#3#4#5{\@tempskipa #3\relax
  \ifdim \@tempskipa>\z@
    \begingroup #4\@hangfrom{\hskip #1}{\interlinepenalty \@M #5\par}\endgroup
  \else \def\@svsechd{#4\hskip #1\relax #5\@punct}\fi
  \@xsect{#3}}

\makeatother

\begin{document}

\maketitle

\thispagestyle{empty}

\begin{abstract}
\noindent
We prove that the redundancy of a $k$-server PIR code of dimension $s$
is $\Omega(\sqrt{s})$ for all $k \ge 3$. This coincides with a known
upper bound of $O(\sqrt{s})$ on the redundancy of PIR codes. Moreover, 
for $k=3$ and $k = 4$, we determine the lowest possible redundancy of 
$k$-server PIR codes exactly. Similar results were proved independently by 
Mary Wootters 
using a different~method. 
\end{abstract}


\vspace{3.00ex}

\noindent
Given two binary vectors $\uuu = (u_1,u_2,\ldots,u_n)$ and 
$\vvv = (v_1,v_2,\ldots,v_n)$, we define their 
\emph{product $\uuu\vvv$} componentwise, namely\vspace{-0.50ex}
\be{product-def}
\uuu\vvv 
\ \ \deff \,\
(u_1 v_1, u_2 v_2, \ldots, u_n v_n)
\ee
where $u_1v_1, u_2 v_2, \ldots, u_nv_n$ are computed in $\mathrm{GF}(2)$.
Note that the product operation in \eq{product-def} distri\-butes
over addition in $\F_2^n$. Thus \eq{product-def} turns the vector
space $\F_2^n$ into an algebra $\cA_n$ over $\F_2$. 
This~algebra $\cA_n$ is unital, associative, and commutative.

Given a set $X \subseteq \F_2^n$, we define the square of $X$
as the set of products of the elements in $X$.~Explicitly, 
$X^2$ is defined as follows:\vspace{-1.00ex}
\be{square-def}
X^2 \,\ \deff\,\
\bigl\{\, \uuu \vvv \,:\, \uuu,\vvv \in X ~\text{and}~ \uuu \ne \vvv\bigr\}
\ee
The following lemmas follow straightforwardly from the
definitions in \eq{product-def} and \eq{square-def}, along
with the fact that 
$\cA_n$ is a commutative algebra.
We let $\Span{X}$ denote the linear span over $\F_2$
of a set $X \subseteq \F_2^n$.

\vspace{0.50ex}
\begin{lemma}
\label{L1}
$|X^2| \,\le\, |X|\bigl(|X|-1\bigr)/2$.
\vspace{-0.75ex}
\end{lemma}

\begin{proof}
If $|X| = r$, then $X^2$ consists of the $\binom{r}{2}$ vectors
$\uuu \vvv = \vvv \uuu$ for some $\uuu \ne \vvv$ in $X$.
Some of these vectors may coincide.
\end{proof}

\begin{lemma}
\label{L3}
Let $\uuu,\vvv_1,\vvv_2,\vvv_3 \in \F_2^n$. \,If\, 
$\vvv_1\vvv_2 + \vvv_1\vvv_3 + \vvv_2\vvv_3 = \zero$,
then
$$
(\uuu + \vvv_1)(\uuu + \vvv_2)
\,+\,
(\uuu + \vvv_2)(\uuu + \vvv_3)
\,+\,
(\uuu + \vvv_3)(\uuu + \vvv_1)
\ = \
\uuu
$$
\vspace{-4.50ex}
\end{lemma}

\begin{proof}
Follows by straightforward verification using distributivity 
and commutativity in $\cA_n$.
\end{proof}

\vspace{1.00ex}
We now show how the foregoing lemmas can be used to establish
a bound on the redundancy of binary $k$-server PIR codes for $k \ge 3$.
These codes are defined in~\cite{FVY15a,FVY15b} 
as follows. 

\begin{definition}
\label{PIRmatrix}
Let $\eee_i$ denote the binary (column) vector with\/ $1$ in position $i$
and zeros elsewhere. We say that an $s \times n$ binary matrix $G$ has\/ 
{\dfn property $\!\cP_k$} if for all $i \in \![s]$, there exist $k$ disjoint 
sets of columns of $G$ that add up to $\eee_i$. A matrix that has 
property $\cP_k$ is also said to be a {\dfn $k$-server PIR matrix}.
A binary linear code\/ $\C$ of length $n$ and dimension $s$ is 
called 
a {\dfn $k$-server PIR code} if there exists a~generator 
matrix $G$ for\, $\C$ with property $\cP_k$.
\end{definition}

For much more on $k$-server PIR codes and their applications
in reducing the storage overhead of private information retrieval,
see~\cite{FVY15a,FVY15b}. 
In particular, it is shown in~\cite{FVY15b} that, given 
a $k$-server PIR code of length $s+r$ and dimension $s$, 
the storage overhead of \emph{any} linear $k$-server PIR protocol
can be reduced from $k$ to $(s+r)/s$. Moreover, for every fixed $k$,
there exist $k$-server PIR codes whose rate (and, hence, storage
overhead) approaches $1$ as their dimension $s$ grows. However,
exactly \emph{how fast} the resulting storage overhead tends to $1$
as $s \to \infty$ was heretofore unknown. For every fixed $k$, 
Fazeli, Vardy, and Yaakobi~\cite{FVY15a,FVY15b} construct $k$-server
PIR codes with redundancy $r$ bounded by
$r \le k \sqrt{s}\bigl(1 + o(1)\bigr)$.
But the question of whether codes with even smaller redundancy
exist was left open in~\cite{FVY15a,FVY15b}. The following 
theorem shows that the redundancy $O(\sqrt{s})$ 
of the codes~con\-structed in~\cite{FVY15a,FVY15b} is asymptotically optimal.

\begin{theorem}
\label{main}
Let\, $\C$ be a $3$-server PIR code of length $n$ and dimension $s$.
Let $r = n-s$ denote the redundancy of\, $\C$. Then\, 
$r(r\,{-}\,1) \ge 2s$.
\end{theorem}

\begin{proof}
Let $G$ be an $s \times n$ generator matrix for $\C$ with property $\cP_3$,
and let $\xxx_1,\xxx_2,\ldots,\xxx_n$ denote the columns of $G$. 
By definition, for each $i \,{\in}\, [s]$, there exist $3$ disjoint subsets
of $\{\xxx_1,\xxx_2,\ldots,\xxx_n\}$ that add up to $\eee_i$. Let
$R_1,R_2,R_3 \subset [n]$ denote the corresponding sets of indices.
Then we can write
\be{thm1-a}
\eee_i 
\; = \ 
{\displaystyle\sum}_{j \in R_1} \xxx_j 
\ = \
{\displaystyle\sum}_{j \in R_2} \xxx_j 
\ = \
{\displaystyle\sum}_{j \in R_3} \xxx_j 
\vspace{0.25ex}
\ee

It is easy to see from \Dref{PIRmatrix} that $G$ has full column rank.
Hence some $s$ columns of $G$ are linearly independent, and we assume
w.l.o.g.\ that these are the first $s$ columns. Consequently, there
exists a nonsingular $s \times s$ matrix $A$ such that 
\be{G'}
G' 
\ \ \deff \:\
AG 
\ = \
\bigl[\,I_s\,|\,P\,\bigr]
\ee
where $I_s$ is the $s \times s$ identity matrix and $P$ is an
$s \times r$ matrix. Let $\xxx'_1,\xxx'_2,\ldots,\xxx'_n$ denote 
the columns of $G'$, with $\xxx'_j = \eee_j$ for $j = 1,2,\ldots,s$.
Then it follows from \eq{thm1-a} that
\be{thm1-b}
\aaa_i 
\; = \ 
{\displaystyle\sum}_{j \in R_1} \xxx'_j 
\ = \
{\displaystyle\sum}_{j \in R_2} \xxx'_j 
\ = \
{\displaystyle\sum}_{j \in R_3} \xxx'_j 
\vspace{0.25ex}
\ee
where $\aaa_1,\aaa_2,\ldots,\aaa_s$ are the columns of $A$.
Note that $\dim \Span{\aaa_1,\aaa_2,\ldots,\aaa_s} = s$,
since the matrix $A$ is nonsingular. Let us now further 
define\vspace{0.75ex}
\begin{align}
\label{S-def}
S_1 &\,=\, R_1 \cap [s], &
~
S_2 &\,=\, R_2 \cap [s], &
~
S_3 &\,=\, R_3 \cap [s]
\\[0.75ex]
\label{T-def}
T_1 &\,=\, R_1 \cap \bigl([n]{\setminus}[s]\bigr), &
~
T_2 &\,=\, R_2 \cap \bigl([n]{\setminus}[s]\bigr), &
~
T_3 &\,=\, R_3 \cap \bigl([n]{\setminus}[s]\bigr)
\\[0.75ex]
\label{v-def}
\vvv_1 &\,=\, \sum_{j\in S_1}\!\xxx'_j \ = \sum_{j\in S_1}\!\eee_j ~~~~&
\vvv_2 &\,=\, \sum_{j\in S_2}\!\xxx'_j \ = \sum_{j\in S_2}\!\eee_j ~~~~&
\vvv_3 &\,=\, \sum_{j\in S_3}\!\xxx'_j \ = \sum_{j\in S_3}\!\eee_j \hspace*{5.00ex}
\\[-2.50ex]
\nonumber
\end{align}
With this notation, we can rewrite \eq{thm1-b} as follows:\vspace{1.00ex}
\be{thm1-c}
\aaa_i + \vvv_1
\ = \ 
{\displaystyle\sum}_{j \in T_1} \xxx'_j 
\hspace{8.00ex}
\aaa_i + \vvv_2
\ = \
{\displaystyle\sum}_{j \in T_2} \xxx'_j 
\hspace{8.00ex}
\aaa_i + \vvv_3
\ = \
{\displaystyle\sum}_{j \in T_3} \xxx'_j 
\hspace*{4.00ex}
\ee
Finally, let us define 
$X \,\ \deff\ \bigl\{\xxx'_{s+1},\xxx'_{s+2},\ldots,\xxx'_{n}\bigr\}$.
Then it follows from~\eq{thm1-c}~that 
$\aaa_i + \vvv_1$, $\aaa_i + \vvv_2$,~and $\aaa_i + \vvv_3$
belong to $\Span{X}$. 
We are now 
ready to use Lemmas \ref{L1} and \ref{L3} in order 
to complete the proof.

Since the sets $S_1,S_2,S_3$ are disjoint, it follows from
\eq{v-def} that the supports of $\vvv_1,\vvv_2,\vvv_3$
are also disjoint. In other words, 
$
\vvv_1 \vvv_2 = \vvv_1 \vvv_3 = \vvv_2 \vvv_3 = \zero
$.
Using \Lref{L3}, we conclude that
\begin{eqnarray*}
\aaa_i
& = &
\hspace*{-0.25ex}(\aaa_i + \vvv_1)(\aaa_i + \vvv_2)
\ + \ 
(\aaa_i + \vvv_2)(\aaa_i + \vvv_3)
\ + \ 
(\aaa_i + \vvv_3)(\aaa_i + \vvv_1)
\\
& = &\hspace*{-0.75ex}%
\left(\sum_{j \in T_1} \xxx'_j\right)\!\left(\sum_{j \in T_2} \xxx'_j\right)
\ + \
\left(\sum_{j \in T_2} \xxx'_j\right)\!\left(\sum_{j \in T_3} \xxx'_j\right)
\ + \
\left(\sum_{j \in T_3} \xxx'_j\right)\!\left(\sum_{j \in T_1} \xxx'_j\right)
\\[1.00ex]
& = &\hspace*{-0.50ex}%
\sum_{j \in T_1}\sum_{k \in T_2} \xxx'_j \xxx'_k\:
\ + \
\sum_{j \in T_2}\sum_{k \in T_3} \xxx'_j \xxx'_k\:
\ + \
\sum_{j \in T_3}\sum_{k \in T_1} \xxx'_j \xxx'_k
\end{eqnarray*}
Since the sets $T_1,T_2,T_3$ are disjoint subsets of
$[n]{\setminus}[s]$, all of the products $\xxx'_j \xxx'_k$
above belong to $X^2$. Consequently, it follows that
\smash{$\aaa_i \in \Span{X^2}$} for all $i$. Hence
$$
\dim {\textstyle\Span{X^2} }
\ \ge \ 
\dim \Span{\aaa_1,\aaa_2,\ldots,\aaa_s} 
\ = \
s
$$
But 
$
\dim {\textstyle\Span{X^2}}
\le
|X^2|
\le
r(r-1)/2
$,
where we have used \Lref{L1}.
Thus 
$r(r-1)/2 \ge s$,~which completes the
proof of the theorem.
\end{proof}

\vspace{1.00ex}
It is shown in~\cite{FVY15a,FVY15b}
that the redundancy of $k$-server PIR codes is non-decreasing
in $k$. That is,~if $\rho(s,k)$ denotes the lowest possible
redundancy of a $k$-server PIR code of dimension $s$,
then
$$
\rho(s,k+1) 
\ \ge \
\rho(s,k) 
\hspace{7ex}
\text{for all $s \ge 1$ and all $k \ge 2$}
$$
Consequently, the lower bound of \Tref{main} trivially
extends from $3$-server PIR codes to general $k$-server PIR codes
with $k \ge 3$.

\looseness=-1
The following simple construction achieves the
lower bound of \Tref{main} for $k = 3$. Let $r$ be the smallest 
integer such that $\binom{r}{2} \ge s$. Take 
$G = \bigl[\,I_s\,|\,P\,\bigr]$, 
where $P$ is an $s \times r$ matrix whose rows are distinct
binary vectors of weight $2$. Clearly, the rows of $P$ form
a constant-weight binary code with 
distance $2$.
By the results of~\cite{FVY15a,FVY15b}, this implies that
$G$ is a $3$-server PIR matrix, and therefore
\be{rho3}
\rho(s,3) 
\ = \
\text{the smallest integer $r$ such that\kern1pt\ $r(r\,{-}\,1) \ge 2s$}
\ = \
\left\lceil \sqrt{2s + \sfourth} ~\,+\: \shalf \right\rceil
\ee
It is also shown in~\cite{FVY15a,FVY15b} that for all even $k$,
we have $\rho(s,k) = \rho(s,k{-}1) + 1$. Consequently, \eq{rho3}
determines the lowest possible redundancy of $4$-server PIR codes
as well.

\end{document}